\documentclass[letterpaper, 10 pt, conference]{ieeeconf}
\IEEEoverridecommandlockouts                              % This command is only
\usepackage{mathrsfs}
\usepackage{amsfonts}
\usepackage{amsmath}
\usepackage{graphicx}
\usepackage{subfigure}
 
\newtheorem{theorem}{Theorem}
\newtheorem{lemma}[theorem]{Lemma}

\newtheorem{proposition}[theorem]{Proposition}

\newtheorem{conjecture}[theorem]{Conjecture}

\newcommand{\comment}[1]{}
\newcommand{\spc}{{\cal S}}
\newcommand{\net}{{\cal N}}

\newcommand{\cS}{\mathcal{S}}
\newcommand{\cX}{\mathcal{X}}
\newcommand{\cG}{\mathcal{G}}

\newcommand{\negspace}{\hspace*{-2mm}}

\title{Some bounds on the capacity of communicating the sum of sources}

\author{
\authorblockN{Brijesh Kumar Rai,  Bikash Kumar Dey and Sagar Shenvi}
\authorblockA{Department of Electrical Engineering \negspace \\
Indian Institute of Technology Bombay\negspace \\
Mumbai, India, 400 076\\
\{bkrai,bikash,sagars\}@ee.iitb.ac.in}
}

\begin{document}

\maketitle

\begin{abstract}
We consider directed acyclic networks with multiple sources and multiple
terminals where each source generates one i.i.d. random process over an abelian
group and all the terminals want to recover the sum of these random processes.
The different source processes are assumed to be independent. The
solvability of such networks has been considered in some previous works.
In this paper we investigate on the capacity of such networks,
referred as {\it sum-networks},
and present some bounds in terms of min-cut, and the numbers of sources
and terminals. 

\end{abstract}

\section{Introduction}

The seminal work by Ahlswede et al. \cite{ahlswede1} started a new
regime of communication in a network where intermediate nodes
are allowed to combine incoming information to construct outgoing
symbols/packets. This has been popularly known as network coding.
It was shown that the capacity of a multicast network under network
coding is the minimum of the min-cuts of the individual terminals
from the source. The multicast
capacity under routing may be strictly less than that with coding.
The area has subsequently seen rapid developments.
Linear coding was proved to be sufficient to achieve capacity of
a multicast network in \cite{li1}. Koetter and M\'{e}dard \cite{koetter2}
proposed a different framework of random and deterministic
linear network coding, and Jaggi et. al \cite{jaggi1} proposed
a polynomial time algorithm for designing a linear network code for
a multicast network. The capacity of networks with routing and network
coding was investigated in \cite{cannons1,dougherty4,langberg2}.

In this paper, we consider a directed acyclic network with multiple
sources and terminals where the sources generate one random process each
and the terminals require the sum of those processes. We call such a
network as a {\it sum-network}. The alphabet
of the source processes is assumed to be a finite abelian group $G$, and the sum
is defined as the operation in $G$. We allow fractional
vector network coding where the number $k$ of sums communicated to
the terminals may be different from the vector dimension $l$. The
capacity is then defined naturally as the suppremum of all rates
$k/l$ which are achievable. When the alphabet is a field or more generally a module
over a commutative ring with identity, the capacity achieved
by using only linear codes over that ring
is referred as the linear coding capacity.

The problem of distributed function computation has been considered
previously in the literature in different flavors (see \cite{gallager2,giridhar,kanoria,korner1,han1,feng1}
for example). In the context of network coding, and along the same
line as our present work, communicating the sum of the sources
has been considered in several past works. Ramamoorthy (\cite{ramamoorthy})
showed that if the number of sources or the number of terminals
is not more than two, then the sum of the sources can be communicated
if and only if each source-terminal pair is connected.
On the other hand, there are networks (\cite{RaiD:09a, langberg3})
with more than two sources and
terminals where the sum can not be communicated at rate one even
though every source-terminal pair is connected.
In \cite{RaiD:09a, RaiD:09b, RaiD:09c}, the authors showed
the richness of this problem as a class by showing existence of
networks which are linearly solvable only over finite fields of characteristics
belonging to a given finite or Co-finite set of primes, existence
of network which is solvably equivalent to any (non-function)
general network coding, and thus equivalent to any given system
of polynomial equations \cite{dougherty2}. It was also shown that
by using a code construction originally given in \cite{koetter3}, any
fractional coding solution of a sum-network also naturally provides
a fractional coding solution of the same rate for the reverse network.
The case of one terminal and more general functions have been
considered in \cite{appuswamy1, appuswamy2}.

In this paper, we consider the problem of communicating the sum of
the sources over a network to a set of terminals and investigate
the capacity of such networks. The exact characterization of the
capacity seems to be difficult and we present some bounds, and find
the capacity exactly for some interesting networks with three sources
and three terminals.

The paper is organized as follows. In Section \ref{sec:model}, we formally
introduce the system model and some preliminary definition.
The results of the paper, that is, the bounds on the capacity
of sum-networks are presented in Section \ref{sec:capacity}. We end
with a discussion in Section \ref{sec:disc}.

\section{System model and definitions}\label{sec:back}
\label{sec:model}
We consider a directed acyclic multigraph $\cG= (V,E)$, where $V$ is a 
finite set of nodes and $E \subseteq V \times V$ is the set of edges 
in the network. For any edge $e=(i,j)\in E$, the node $j$ is called 
the head of the edge and the node $i$ is called the tail of the edge; 
and are denoted as $head (e)$ and $tail (e)$ respectively. For each node $v$, 
$In(v)= \{e \in E \colon head (e) = v \}$ is the set of incoming edges at the node $v$. 
Similarly, $Out(v)= \{e \in E \colon tail (e) = v \}$ is the 
set of outgoing edges from the node $v$. 

Each edge in the network is capable of carrying 
a symbol from the alphabet in each use. Each edge is used once per unit 
time and is assumed to be zero-error and zero-delay communication 
channel. A network code is an assignment of an edge 
function to each edge and a decoding function to each terminal. 
In a  $(k,l)$ fractional network code, $k$ symbols generated at each source 
are blocked and encoded into $l$-length vectors on the outgoing edges.
All the internal edges also carry $l$-length vectors.
Thus
%The ratio $k/l$ is the {\it rate} of the $(k,l)$ fractional network code. 
%A rate $k/l$ is said to be achievable if there is a $(k,l)$
%fractional solution for the network. The suppremum of all achievable
%rates is defined to be the {\it capacity}.
%
%For any edge $e\in E$, $Y_e \in F^l$ denotes the symbol transmitted
%through $e$ and for a terminal node $v$, $R_v \in F^k$ denotes the
%symbol recovered by the terminal $v$. 
%
for a $(k,l)$ fractional network code over $G$, an edge function for
an edge $e$, with $tail (e) = v$, is defined as
\begin{eqnarray}
f_{e} \colon G^{k}\rightarrow G^{l}, \mbox{ if } v \in S \label{code1}
\end{eqnarray}
and
\begin{eqnarray}
f_{e} \colon G^{l|In(v)|}\rightarrow G^{l}, \mbox{ if } v \notin S. \label{code2}
\end{eqnarray}

A decoding function for a terminal $v$ is defined as 
\begin{eqnarray}
g_{v} \colon G^{l|In(v)|}\rightarrow G^{k}. \label{code3}
\end{eqnarray}
The goal in a
sum-network is that the terminals should be able to recover the sum of
the $k$-length vectors generated at the sources.
A $(k,l)$ fractional network code over $G$ is called a 
$k$-length or a $k$-dimension {\it vector network code} over $G$ 
if $k=l$ and called a {\it scalar network code} over $G$ if $k=l=1$. 
A network code is called a {\it linear network code} when the alphabet
is a field, or more generally a module over a commutative ring with identity,
and all the edge functions and the decoding functions are
linear functions over the alphabet field or the ring.
Note that even if the alphabet is an abelian group, one can talk about
a linear solution by considering the abelian group as a module over
the integer ring and then a code can be linear over the integer ring.

A network has a solution over $G$ using a $(k,l)$ fractional network code over $G$ 
if the demand of each terminal node is fulfilled using some $(k,l)$ fractional network 
code over $G$. 
The ratio $k/l$ is the {\it rate} of the $(k,l)$ fractional network code.
A rate $k/l$ is said to be {\it achievable} if there is a $(k,l)$
fractional solution for the network. The suppremum of all achievable
rates is defined to be the {\it capacity}.
The {\it linear network coding capacity} of a network is the suppremum of
all rates that are achievable using fractional linear network codes. 
Clearly, the network coding capacity of a network 
is greater than or equal to the linear network coding capacity.
A sum-network is said to be solvable (resp. linearly solvable) if it has a $(1,1)$
coding (resp. linear coding) solution.

\section{Capacity of sum-networks}
\label{sec:capacity}
First, we mention the following simple lower bound on the capacity of any
sum-network.
\begin{theorem}
The capacity of a sum-network is bounded by the minimum of the min-cuts
of all source-terminal pairs. That is,
\begin{eqnarray}
Capacity & \leq &  min_{i,j}(\mbox{\emph{min-cut}} \ (s_i-t_j)). \nonumber 
\end{eqnarray}
\end{theorem}
\begin{proof}
For any source $s_i$ of the network, let us fix the source processes
of the other sources to the all-zero ("zero" being the identity
element of the alphabet group) sequence. Then the problem reduces to
the multicast problem from the source $s_i$ to all the terminals, and
the capacity of this problem is the minimum of the min-cuts from $S_i$
to all the terminals. The overall capacity of the sum-network must be
less than or equal to each of these multicast capacities for
different $i$.
\end{proof}
In \cite{RaiD:09b, RaiD:09c}, the reverse of a sum-network was considered
where the direction of the edges are reversed and the role of sources and
terminals is interchanged. It was shown, using a code-construction originally
described in a basic form in \cite{koetter3} as the {\it dual code} in the language of
codes on graphs, that if a sum-network has a $(k,n)$ fractional
linear solution, then from such a network code, one can also construct
a $(k,n)$ fractional linear solution of the reverse sum-network.
This means that the linear coding capacity of the reverse sum-network is
the same as the linear coding capacity of the original sum-network.

Our lower bounds on the capacity of sum-networks have varying degree of
tightness depending on the number of sources and the number of terminals
of the network. So we present these bounds in different subsections
dealing with various numbers of sources and terminals.
For the rest of the paper, $m$ and $n$ will denote the number of sources
and the number of terminals respectively.

\subsection{The case of $\min \{m,n\}=1$}
If the sum-network has only one source, then the network is a multicast
network. The capacity of a multicast network is known to be {\it equal}
to the minimum of the min-cuts of the source-terminal pairs, and thus
the capacity achieves the min-cut upper bound. Moreover, this capacity
is achieved by linear codes if alphabet is a finite field.
Now, over a finite field, for the case of $n=1$, let us consider the reverse network of
a sum-network obtained by reversing the direction of the edges
and interchanging the role of the sources and the terminals.
The reverse network is a multicast network and thus has linear coding
capacity equal to the minimum of the min-cuts of the source-terminal
pairs. So the linear coding capacity of the original one-terminal
sum-network is the minimum of the min-cuts of the source-terminal pairs.
Since the coding capacity is also upper bounded by the min-cut, the
coding capacity of a one-terminal sum-network is the minimum of the
min-cuts of the source-terminal pairs. So we have

\begin{theorem}
The capacity (and the linear coding capacity) of a one-source or
one-terminal sum-network is the minimum of the min-cuts of all
source-terminal pairs.
\end{theorem}

\subsection{The case of $\min \{m,n\}=2$}
\label{subsec:min2}
It was proved in \cite{ramamoorthy} that for a network with $\min \{m,n\}=2$
where every source-terminal pair is connected, it is possible to communicate
the sum of the sources to the terminals (at rate $1$). Which means that
for $\min \{m,n\}=2$, $Capacity \geq 1$ if the minimum of the min-cuts
of the source-terminal pairs is at least $1$. So, the min-cut bound is
tight in this case if the min-cut is $1$. However, if the min-cut is greater
than $1$, then it is not known if this upper bound is achievable.
However, we can always achieve the half of the min-cut
upper bound by time-sharing. For example, for $m=2$, each source can communicate
its symbols in one slot at the rate of min-cut and then after two
time-slots, the terminals can add the symbols received from the two
sources. So, we have

\begin{theorem}
For $\min \{m,n\} = 2$, the capacity of a sum-network is bounded as
\begin{eqnarray}
Capacity & \geq & \max \left\{ \min \{1, min_{i,j}(\mbox{\emph{min-cut}} \ (s_i-t_j))\}, \right. \nonumber \\
&& \hspace{9mm} \left. 0.5\times min_{i,j}(\mbox{\emph{min-cut}} \ (s_i-t_j))\right\}.  \nonumber
\end{eqnarray}
\end{theorem}

\subsection{The case of $m=n=3$}
The case of $m=n=3$ is intriguing. On one hand, these are the smallest
values of $m,n$ for which there is a network (called $\cS_3$ in
\cite{RaiD:09a} and shown in Fig. \ref{fig:S3}) where every source-terminal
pair is connected, i.e., which has min-cut $\geq$ 1, but still does not
have a linear \cite{RaiD:09a} or non-linear \cite{langberg3} solution
of rate $1$. So, these are the smallest parameters for which the min-cut
upper bound is known to be not achievable. (Though it is still not clear
at this point if the min-cut upper bound may still be achievable in the
limit as the suppremum of achievable rates.) On the other hand, from elaborate
investigation of possible networks with these parameters, there seems
to be very limited types of networks. The $S_3$ and its extensions
(essentially the network shown in Fig. \ref{fig:S3a})
seem to be the only "non-solvable" sum-networks for $m=n=3$. The network
(let us call it $\cX_3$) shown in Fig. \ref{fig:X3} was presented in \cite{RaiD:09b}
and was shown to be solvable by scalar linear code over all fields except
the binary field $F_2$.

\vspace*{-4mm}
\begin{center}
\begin{figure}[h]
\centering\includegraphics[width=2.0in]{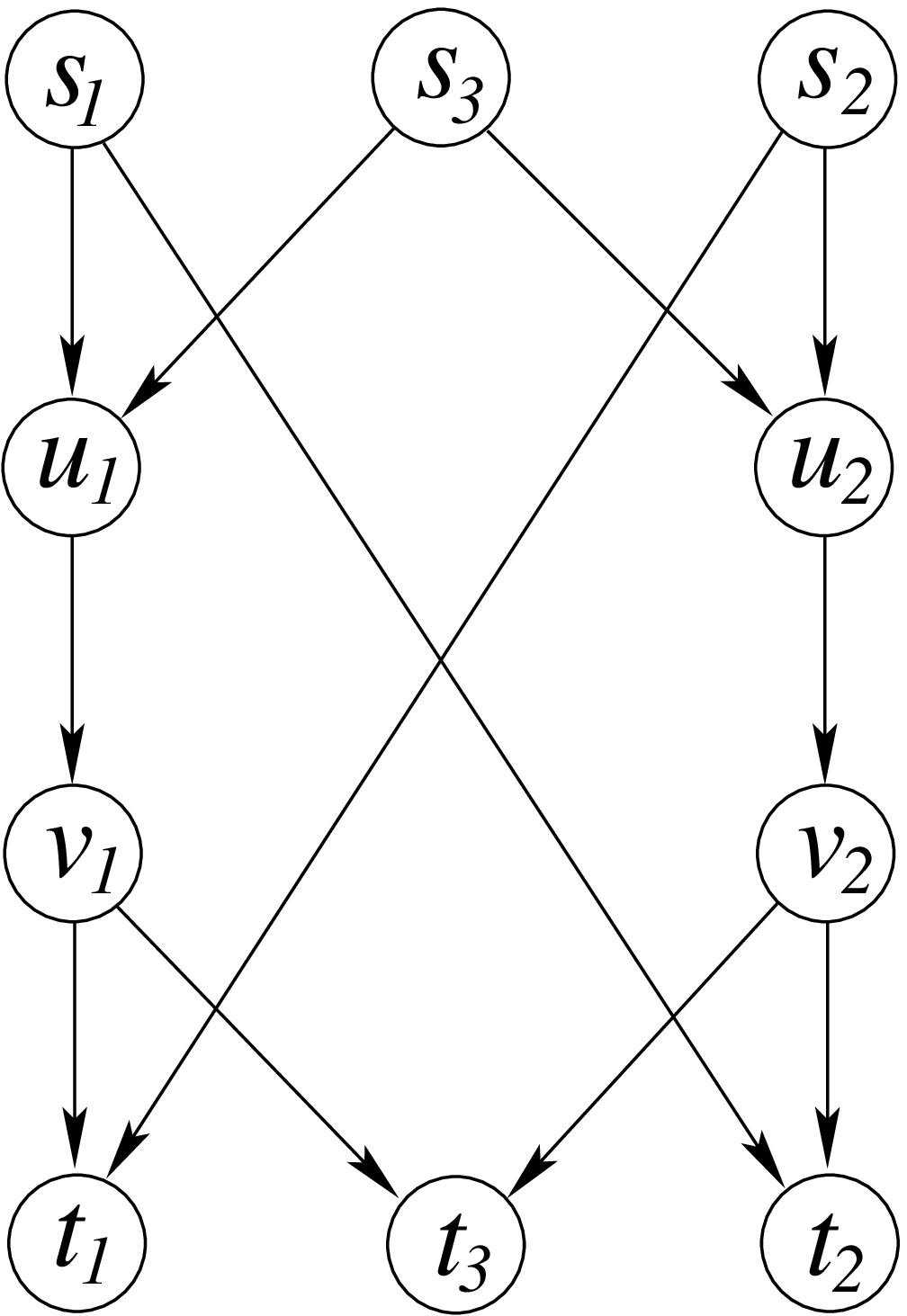} \caption{The
network $\cS_3$} \label{fig:S3}
\end{figure}
\end{center}

\vspace*{-4mm}
\begin{center}
\begin{figure}[h]
\centering\includegraphics[width=2.0in]{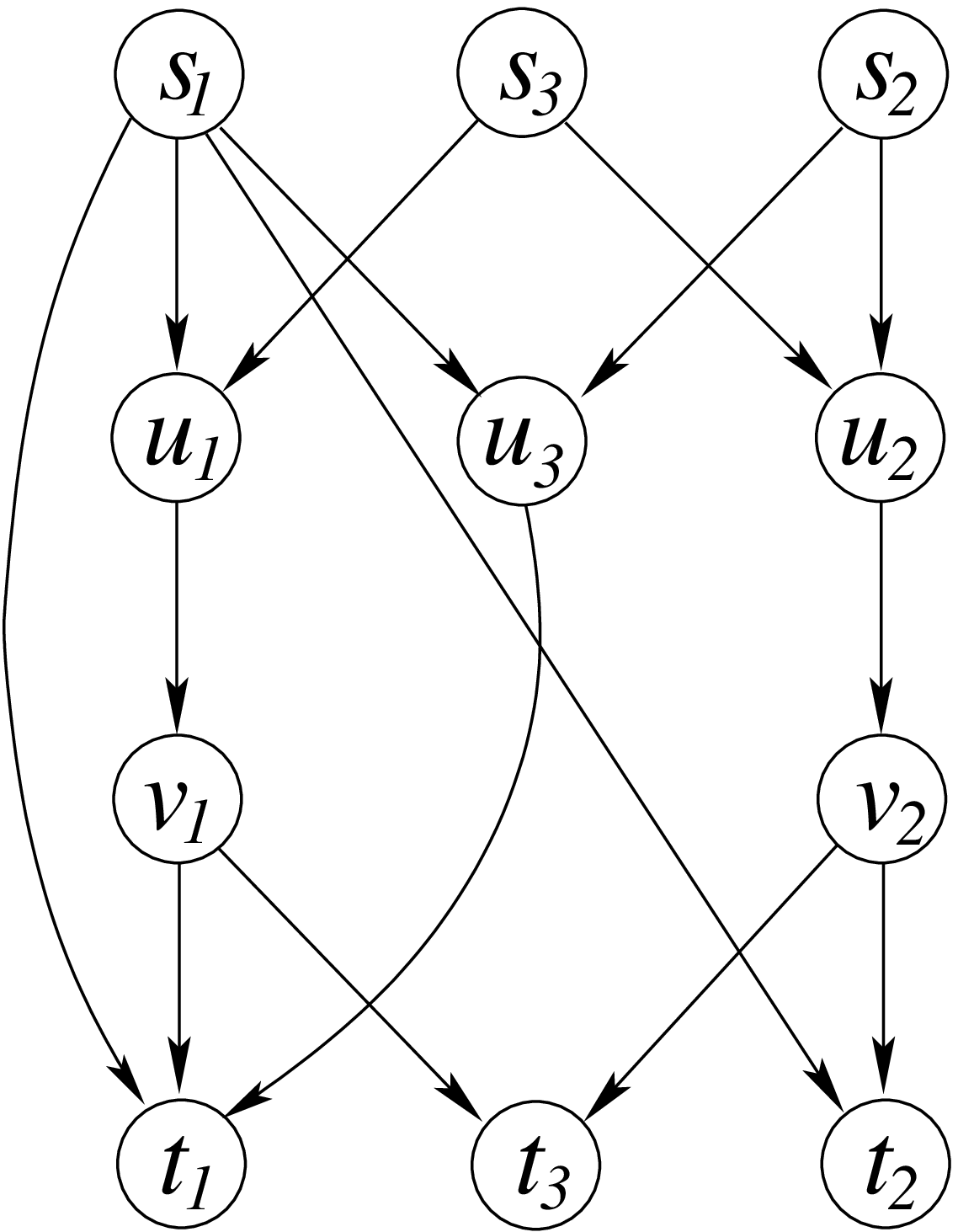} \caption{The
network $\cS_3^\prime$} \label{fig:S3a}
\end{figure}
\end{center}

\vspace*{-4mm}
\begin{center}
\begin{figure}[h]
\centering\includegraphics[width=2.0in]{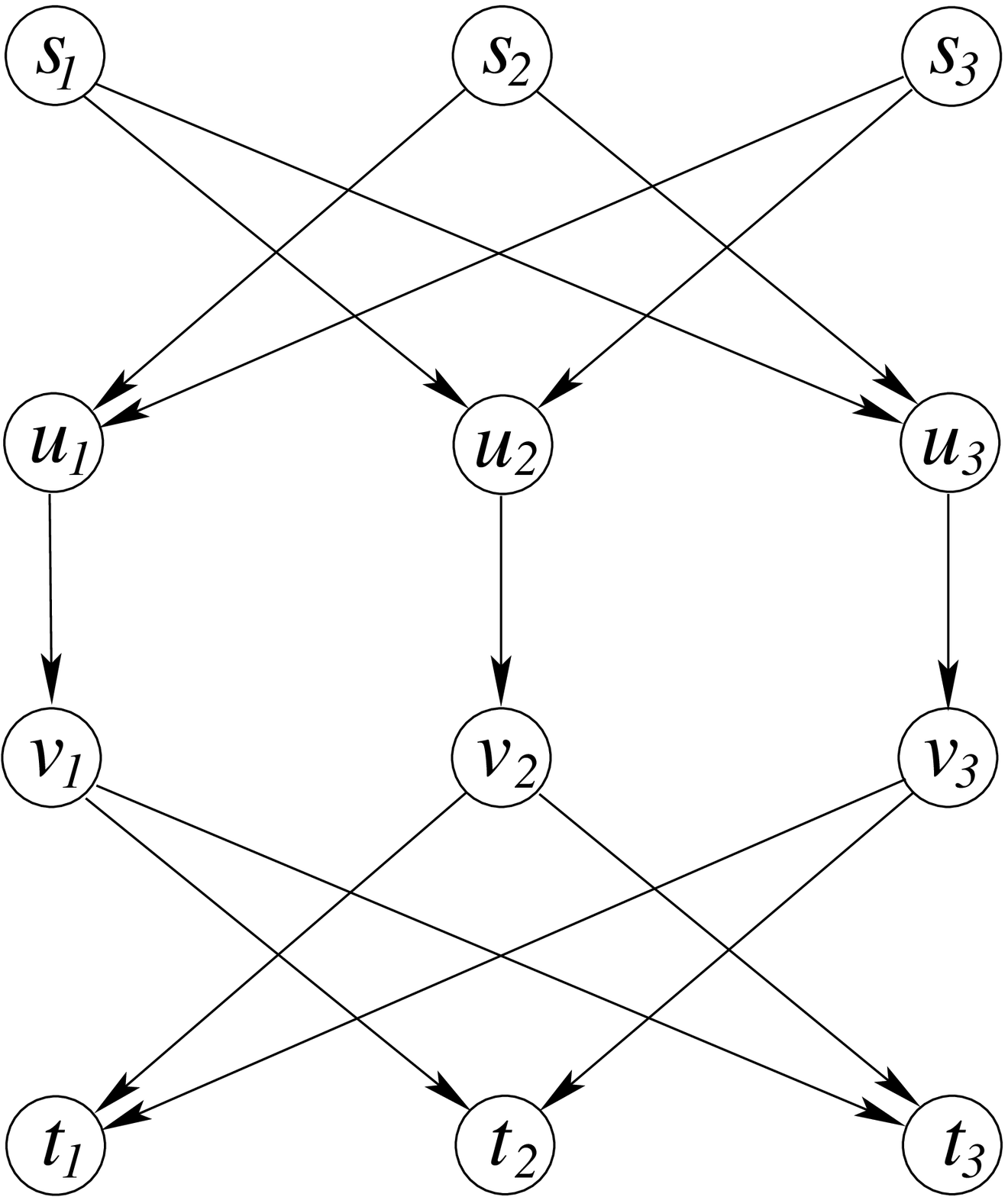} \caption{The
network $\cX_3$} \label{fig:X3}
\end{figure}
\end{center}

%First, we show in the following that the network $\cS_3$ and its extension
%$\cS_3a$ shown in Fig. \ref{} both have capacity exactly $2/3$ whereas
%their min-cut upper bound is $1$. So, there is a gap between the capacity
%and the min-cut upper bound.
First, we give a generic lower bound on the capacity of any sum-network
with $\min \{m,n\} = 3$ with min-cut $\geq 1$.

\begin{theorem}
\label{th:23achivability}
%The capacity of both the networks $\cS_3$ and $\cS_3^\prime$ is $2/3$.
The linear coding capacity of any sum-network with $\min \{m,n\} = 3$ with
min-cut $\geq 1$ is at least $2/3$.
\end{theorem}
\begin{proof}
Without loss of generality, let us assume that the number of terminals
is $3$ (otherwise consider the reverse network).
Let us consider two symbols at each source:
$X_{i1}, X_{i2}$ at $s_i$ for $i=1,2,\ldots, m$.
Let the two sums be denoted as $Sum_1 = \sum_{i=1}^m X_{i1}$.
and $Sum_2 = \sum_{i=1}^m X_{i2}$. If we take two terminals at a time,
the resulting network has a capacity $\geq 1$ as discussed in Section \ref{subsec:min2}
using scalar linear network coding as proposed in \cite{ramamoorthy}.
Now, the two sums $Sum_1$ and $Sum_2$ can be communicated to all the terminals
in three time slots. In the first time slot, $Sum_1$ is communicated
to $t_1$ and $t_2$. In the second time slot, $Sum_2$ is communicated to
$t_2$ and $t_3$. In the third time slot, $Sum_1 + Sum_2 
= \sum_{i=1}^m (X_{i1} + X_{i2})$ is communicated to $t_1$ and $t_3$.
Having received $Sum_1$ (respectively $Sum_2$) and $Sum_1 + Sum_2$, the
terminal $t_1$ (respectively $t_3$) can recover $Sum_2$ (respectively $Sum_1$)
as well. So all the terminals recover the two sums in three time slots, thus
achieving a rate $2/3$ using linear coding.
\end{proof}

It was proved in \cite{langberg3} that if a sum-network with $m=n=3$
has two edge-disjoint paths between any source-terminal pairs, then the
network is linearly solvable, that is, rate $1$ is achievable by
scalar linear coding. This gives the following bound.
\begin{proposition}
The linear coding capacity of any sum-network with $m=n=3$ with
min-cut $\geq 2$ is at least $1$.
\end{proposition}

Now we show that the network $\cS_3$ and its extension
$\cS_3^\prime$ shown in Fig. \ref{fig:S3a} both have capacity exactly $2/3$ whereas
their min-cut upper bound is $1$. So, there is a gap between the capacity
and the min-cut upper bound.

\begin{theorem}
The capacity and linear coding capacity of $\cS_3$ and $\cS_3^\prime$ is
$2/3$.
\end{theorem}
\begin{proof}
Clearly, the network $\cS_3^\prime$ is obtained from $\cS_3$ by adding one
direct edge from $s_1$ to $t_1$, and subdividing the edge $(s_2, t_1)$
and adding one edge into it from $s_1$. So, the capacity
and the linear coding capacity of the network $\cS_3^\prime$ is at least that
of $\cS_3$. By the previous theorem, the rate $2/3$ is achievable by linear
network coding in $\cS_3$. Now we will show that the
capacity of $\cS_3^\prime$ is bounded from above by $2/3$. This will prove that
both the networks have the same capacity and linear coding capacity
and that these are both $2/3$.

Consider any $(k,l)$ fractional network coding solution of $\cS_3^\prime$.
Let $X_1, X_2, X_3 \in G^k$ be the message blocks generated at
the three sources. Let the edges $(u_1, v_1)$ and $(u_2, v_2)$ carry
the functions $f(X_1, X_3)$and $g(X_2, X_3)$ respectively.
For any fixed values of $X_1$ and $X_2$, the set of messages received
by the terminal $t_1$ should be a one-one function of $X_3$ since
the terminal can recover the sum $X_1+X_2+X_3$ which is a one-one function
of $X_3$. Since the messages on $(s_1, t_1)$ and $(u_3, t_1)$
are fixed by the values of $X_1$ and $X_2$, the message on $(v_1,t_1)$ and
thus $f(X_1, X_3)$ must be a one-one function of $X_3$ for a fixed value
of $X_1$.

Now clearly, for $t_2$ to be able to recover the sum, the function
$g$ should be such that one can recover $X_2+X_3$ from $g(X_2,X_3)$.
Since $t_2$ recovers the sum $X_1+X_2+X_3$, and it can recover
$X_2+X_3$ from the message in $(v_2, t_3)$, it can also recover
$X_1$ by subtracting. Now, $t_3$ receives $f(X_1, X_3)$ on
$(v_1, t_3)$ (WLOG) and this is a one-one function of $X_3$ for any
given $X_1$. So, having recovered $X_1$, $t_3$ can recover $X_3$ from
$f(X_1, X_3)$. Then by using $g(X_2, X_3)$ received on $(v_2, t_3)$
and the value of $X_3$, $t_3$ can also recover $X_2$. So, $t_3$
can recover all the original messages $X_1, X_2, X_3$. Now $(X_1, X_2, X_3)$
takes a total of $|G|^{3k}$ possible values as a triple. On the other hand
$\{(u_1, v_1), (u_2, v_2)\}$ is a cut between the sources and $t_3$,
and this cut can carry at most $|G|^{2l}$ possible different message-pairs.
So, we have $|G|^{2l} \geq |G|^{3k} \Rightarrow k/l \leq 2/3$.
\end{proof}

The following observations lead us to believe that the network $\cS_3^\prime$
is essentially the only maximal extension of $\cS_3$ which has the same
capacity.

\begin{enumerate}
\item Further subdividing $(u_3,t_1)$ and adding
an edge from it to $t_2$ makes the network $\cX_3$ a subgraph of the
resulting network, and thus the capacity of the network increases to $1$.

\item Also subdividing $(s_1,t_2)$ and adding
an edge from it into $t_1$ does not change its capacity since there
is already an edge $(s_1, t_1)$ and the new edge can not carry any extra
information to $t_1$. 

\item Also subdividing $(s_1,t_2)$ and adding
an edge to it from $s_2$ gives a strictly richer network than $\cX_3$,
and thus the capacity of the network increases to $1$.

\item Instead of the edge $(s_1, t_1)$, if an edge $(s_2, t_2)$ is added,
then the resulting network (shown in Fig. \ref{fig:X3a}) is strictly richer than $\cX_3$ because the
edges $(s_1, t_2)$ and $(s_2, t_2)$ can jointly carry more information
than an edge from $u_3$ to $t_2$. So the resulting network has capacity
$1$ even though it does not have a binary scalar solution (like $\cX_3$).
\end{enumerate}

These observations also lead us to believe that

\begin{conjecture}
The capacity of a sum-network with $m=n=3$ is either $0, 2/3$ or at least $1$.
\end{conjecture}

\vspace*{-4mm}
\begin{center}
\begin{figure}[h]
\centering\includegraphics[width=2.0in]{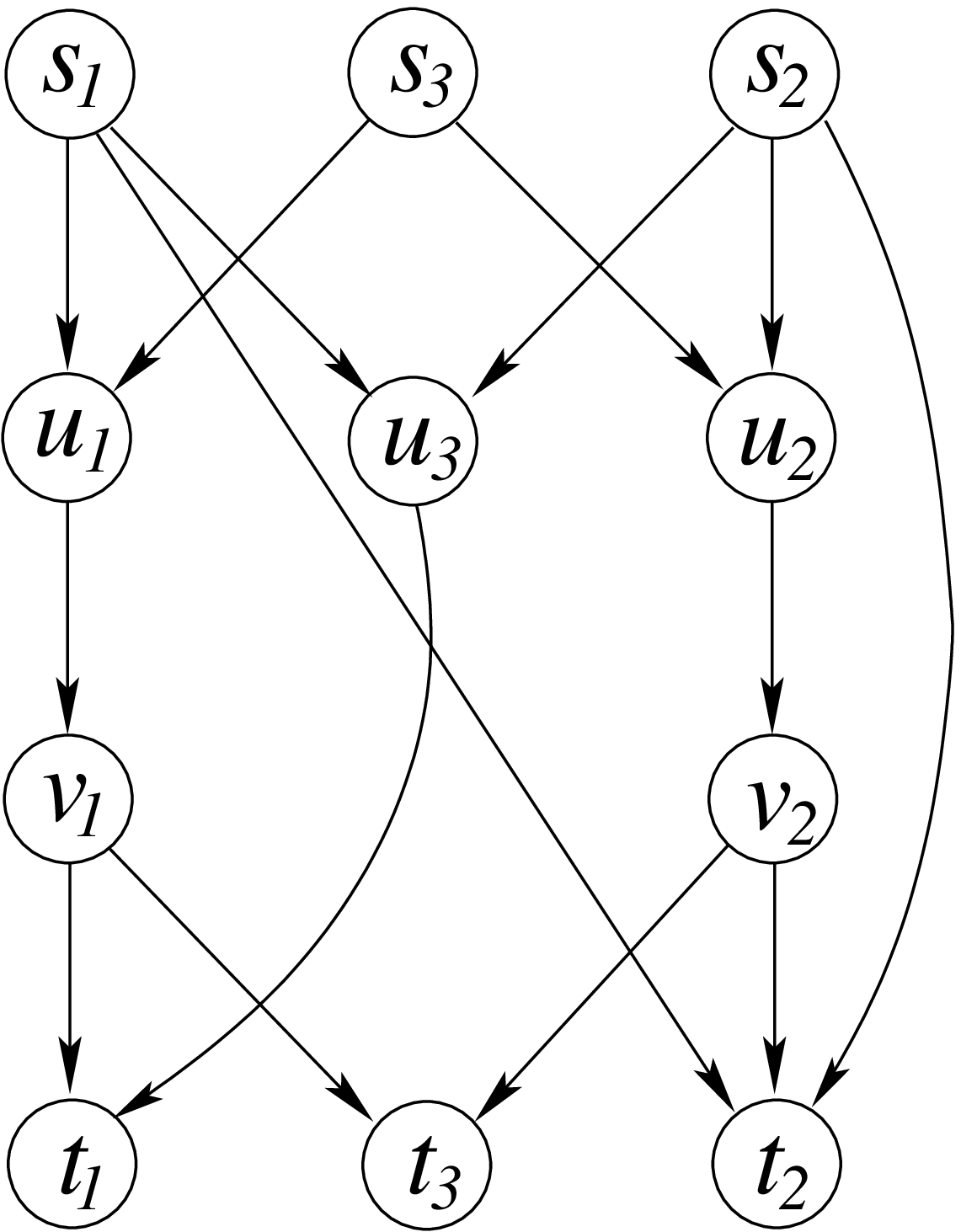} \caption{The
network $\cX_3^\prime$} \label{fig:X3a}
\end{figure}
\end{center}

\begin{table*}[t]
\begin{center}
\begin{tabular}{|l||l|l||l|l|}
\hline\hline
 & \multicolumn{2}{|c||}{Solvability} & \multicolumn{2}{|c|}{Capacity} \\ \cline{2-5}
 & min-cut $= 1$ & min-cut $> 1$ & min-cut $=1$ & min-cut $>1$ \\\hline\hline
$\min\{m,n\}=1$ & Solvable & Solvable & $=1$ & = min-cut \\\hline
$\min\{m,n\}=2$ & Solvable & Solvable & $=1$ & $\geq \text{min-cut}/2$ (loose/tight?) \\\hline
$m=n=3$ & Network dependent & Solvable & $\geq 2/3$ (tight) & $\geq \max\{1, \text{min-cut}/3\}$ (loose!) \\\hline
$\min\{m,n\}=3$ & Network dependent & ? & $\geq 2/3$ (tight) & $\geq \text{min-cut}/3$ (loose!) \\\hline
$\min\{m,n\}>3$ & Network dependent & ? & $\geq 2/\min\{m,n\}$ (loose!) & $\geq \text{min-cut}/\min\{m,n\}$ (loose!) \\\hline\hline
\end{tabular}
\caption{Solvability and bounds on the capacity}
\label{table:bounds}
\end{center}
\end{table*}

\subsection{The case of $m,n > 3$}
This is the most ill-understood class of sum-networks. We only have
what we suspect to be a very loose lower bound on the capacity of this
class of networks. This lower bound is obtained by similar coding
by time-sharing scheme as in the proof of Theorem \ref{th:23achivability}.
\begin{theorem}
The linear coding capacity of a sum-network with $\min\{m,n\}\geq 2$ and
min-cut $\geq 1$ is at least $2/\min\{m,n\}$.
\end{theorem}
\begin{proof}
The case of $\min\{m,n\} = 2$ follows from Theorem \ref{th:23achivability}. For
$\min\{m,n\} > 2$, without loss of generality, let us assume that $n\leq m$.
For even $n$, we can group the terminals into $n/2$ pairs and in each
time slot communicate the sum of the source symbols to one pair of
terminals. So, in $n/2$ time slots we can communicate one sum of the source
symbols to all the terminals thus achieving a rate $2/n$. For odd $n$,
we can group the terminals into $(n-3)/2$ pairs and one triple. We can
communicate one sum to each pair in one time slot. So, we can communicate
two sums to all the pairs in $(n-3)$ slots. Then using the same
scheme as in the proof of Theorem \ref{th:23achivability},
we can communicate two sum to
the group of three terminals in three time slots. So, in overall
$n$ time slots, we can communicate two sums to all the terminals
in the network. This gives us a rate $2/n = 2/\min\{m,n\}$.
\end{proof}

We believe that this bound is very loose and there is scope for improvement.
Even though we failed to come up with any achievable scheme of
higher rate, we also failed to construct a network satisfying the bound
with equality.

\section{Discussion}
\label{sec:disc}
Some upper and lower bounds on the capacity of communicating the sum of
sources to a set of terminals are presented in this paper. A decreasing degree
of tightness is observed or suspected in these bounds as the numbers of
sources and terminals increase. We summarize the bounds in Table
\ref{table:bounds} to bring out this observation. The parenthetic
comments in the table entries indicate the tightness of the bound
as known or conjectured (indicated with an exclamation (!) mark.)
The interrogation (?) mark as an entry indicates that nothing is known
about the case.

\section{Acknowledgment}
The work of B.~K.~Rai was supported in part by Tata Teleservices IIT Bombay
Center of Excellence in Telecomm (TICET). The work of B.~K.~Dey and S.~Shenvi
was supported in part by Tata Teleservices IIT Bombay
Center of Excellence in Telecomm (TICET) and Bharti Centre for Communication.
The authors would like to thank Tony Jacob for fruitful discussions.
\bibliographystyle{unsrt}
\bibliography{ref}
\end{document}